%% file: main.tex
\documentclass[sigconf, balance=false]{acmart}

\setcopyright{none}
\acmDOI{}
\acmISBN{}
\renewcommand\footnotetextcopyrightpermission[1]{}  

\usepackage{microtype}
\usepackage{needspace}
\usepackage{colortbl}
\usepackage{tikz}
\usetikzlibrary{arrows.meta, positioning, shapes, calc, decorations.pathreplacing}
\usepackage{multirow}
\usepackage{array}
\usepackage{xspace}
\usepackage{pifont}
\usepackage{enumitem}
\usepackage{amsthm}
\usepackage{algorithm}
\usepackage{algpseudocode}
\definecolor{caseblue}{HTML}{0055A4}
\newcommand{\caseaccent}[1]{\textcolor{caseblue}{#1}}

\newcommand{\Qphi}{Q_\phi}
\newcommand{\Qphizero}{Q_{\phi^{(0)}}}
\newcommand{\su}{s_u}

\input{paper_revision_macros}

\title[\name{}: Learning Recommendation Queries from Retrieval Feedback]{\name{}: Learning Recommendation Queries\texorpdfstring{\\}{ }from Retrieval Feedback}

\author{Weixin Chen\textsuperscript{1}, Shuchang Liu\textsuperscript{1$\dagger$}, Hailan Yang\textsuperscript{1}, Yuhan Zhao\textsuperscript{1,2}, Li Kang\textsuperscript{2}, Xuefeng Zhang\textsuperscript{1}, Chenghao Zhang\textsuperscript{1}, Xinyuan Wang\textsuperscript{1}, Hanchen Yang\textsuperscript{1}, Yinghao Wu\textsuperscript{1}, Hua Huang\textsuperscript{1}, Yandong Bai\textsuperscript{1}, Fan Yang\textsuperscript{1}, Xiang Li\textsuperscript{1}, Chunjie Chen\textsuperscript{1}, Li Chen\textsuperscript{2}}
\affiliation{%
  \institution{\textsuperscript{1}Kuaishou Technology \quad \textsuperscript{2}Hong Kong Baptist University}
  \country{}
}

\renewcommand{\shortauthors}{Chen et al.}

\begin{document}

\begin{abstract}
Query-based recommendation separates expressing user needs from locating matching items.
Its central learning problem is how to turn item interactions into effective retrieval requests when interaction logs do not annotate queries.
We present \name{}, which uses a fixed retrieval system to construct query supervision and evaluate recommendation outcomes.
Retrieval verification first selects generated expressions that recover observed items as supervised fine tuning targets, allowing retrieval behavior to determine the target wording.
Policy refinement then evaluates generated query sets by the observed item's rank in their fused recommendation list.
Log-rank feedback distinguishes outcomes beyond the recommendation cutoff, while a KL penalty constrains policy change relative to the supervised reference.
The stages have distinct roles: verification identifies expressions to imitate, and policy refinement improves the combined ranking produced by newly generated queries.
Only the query generator is trained; the retriever and item index remain fixed, and no query annotations or learned reward model are required.
Across three public datasets, \name{} achieves state-of-the-art overall recommendation performance among the evaluated methods.
Extensive experiments demonstrate the effectiveness of the proposed design and its generalization across generators and domains.
\end{abstract}


\maketitle
\let\thefootnote\relax\footnotetext{\textsuperscript{$\dagger$}Corresponding author: liushuchang@kuaishou.com}

\input{sections/01_intro}
\input{sections/02_prelim}
\input{sections/03_framework}

\input{sections/05_empirical}

\input{sections/06_related}
\input{sections/07_discussion}

\Needspace{10\baselineskip}
\bibliographystyle{ACM-Reference-Format}
\bibliography{bib/refs}

\clearpage
\appendix
\input{sections/08_appendix}
\input{sections/09_theory}

\end{document}

%% file: paper_revision_macros.tex
\newcommand{\name}{R2QRec\xspace}
\newcommand{\papergraphic}[2]{%
  \IfFileExists{#2.name.tex}{\begingroup\def\namedwidth{#1}\input{#2.name.tex}\endgroup}{\includegraphics[width=#1]{#2}}}

%% file: sections/01_intro.tex
\begin{figure}[t]
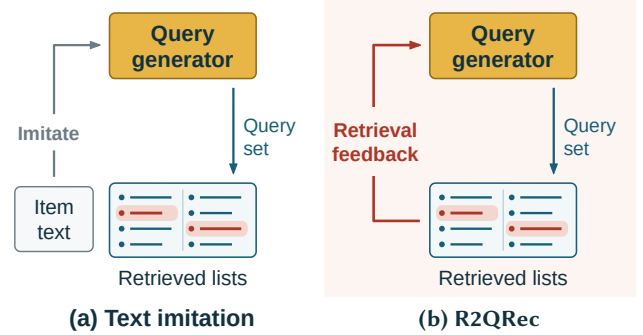

\centering
\papergraphic{\columnwidth}{figures/R2QRec_Figure1.pdf}
\caption{\textbf{Learning to generate recommendation queries.}
(a) Text imitation supervises the generator with an observed item's text.
(b) \name{} uses the observed item's retrieval outcomes to verify queries and optimize the generator.
Red denotes retrieval feedback and the observed item.}
\label{fig:query_learning}
\Description{Two panels share a flat query generator and multiple retrieved lists connected by a query set arrow. Text imitation learns from item text independently of the list. \name{} closes the learning loop from retrieved items to the generator through query verification and rank optimization. Both lists show the observed item at the same illustrative position.}
\end{figure}

\section{Introduction}
\label{sec:intro}

Recommendation on Web platforms requires finding items that satisfy a user's needs, even when those needs do not specify a concrete retrieval request.
Interaction histories and other available user context provide evidence of preferences, inducing several plausible directions for the next recommendation.
Existing approaches express these preferences through discriminative item scores~\citep{kang2018sasrec,he2020lightgcn} or distributions on discrete item identifiers~\citep{rajput2023tiger,onerec2024}.

Query-based recommendation makes this intermediate preference expression explicit in natural language, then uses a retriever to find matching items across the catalog~\citep{li2023gpt4rec,han2025querec}.
A single user can have multiple interests~\citep{li2019mind}.
For example, a reading history containing \emph{Thinking, Fast and Slow}, \emph{The Power of Habit}, and \emph{Outliers} may suggest interests in cognitive science, behavioral change, or human achievement.
Large language models (LLMs) can draw on semantic knowledge to turn these plausible interests into separate retrieval requests, with prior work demonstrating the value of multiple queries for recommendation~\citep{li2023gpt4rec}.
This provides an open-ended language representation of preferences without requiring the generator to select an item code or enumerate a candidate list.

Title-supervised and prompting-based recommenders provide two ways to obtain retrieval text.
GPT4Rec trains its generator to predict target item titles before using the generated text for retrieval~\citep{li2023gpt4rec}; BIGRec also uses target item text for supervised adaptation~\citep{bao2023bigrec}.
QueRec~\citep{han2025querec} generates personalized queries through prompting without LLM fine-tuning.
We study a complementary learning formulation: next-item interactions supply query supervision and evaluate the recommendation outcomes of generated queries (Figure~\ref{fig:query_learning}).

An observed interaction identifies an item to recover, but does not prescribe a unique query to reproduce.
Several expressions may retrieve the same item, while the interaction log does not annotate which expression to imitate.
We address this supervision problem through \emph{item-verified supervision}: generated expressions that retrieve the observed item successfully become training targets.
The item label thereby specifies an admissibility criterion for queries without fixing their wording.

Admissible individual expressions and useful query combinations pose different ranking problems.
Even the complete vector of target ranks across queries does not determine the fused recommendation outcome: the positions of competing items across the lists also matter.
Our analysis formalizes this distinction, motivating \emph{fusion-aware refinement} of generated sets using the observed item's rank after reciprocal rank fusion.
Log-rank feedback then supplies relative comparisons among sets for the same context~\citep{shao2024deepseekmath}.
Its cutoff-risk decomposition explains how it emphasizes improvements near the head while retaining feedback beyond a fixed recommendation cutoff.

These choices define \textbf{\name{}}, a compact framework linking item-level supervision to recommendation-level evaluation.
Item-verified supervision constructs the expressions to imitate; fusion-aware refinement evaluates newly generated combinations; the supervised reference constrains policy change through a KL penalty.
The same frozen retrieval interface supplies both supervision and outcome evaluation, while only the query generator is updated.
This makes recommendation queries a learned interface between user context and item retrieval.

At inference, the generator expresses user preferences as several queries, the retriever executes them, and reciprocal rank fusion combines the resulting item lists~\citep{cormack2009rrf}.
The item index and retrieval model remain fixed, and the learning procedure requires neither annotated queries nor a separate learned reward model.

We evaluate \name{} on three public recommendation datasets and observe state-of-the-art overall recommendation performance among the evaluated methods.
Controlled MovieTV comparisons trace title supervision and retrieval-verified supervision through SFT and subsequent policy refinement.
These comparisons connect the design choices to their effect on recommendation quality; further experiments examine multiple queries, model choices, and domain transfer.

Our contributions are summarized as follows:
\begin{itemize}[leftmargin=*, noitemsep, topsep=3pt]
    \item \textbf{Item-supervised query learning.}
We formulate query supervision through the retrieval admissibility of generated expressions, allowing an observed interaction to supervise multiple phrasings without query annotations.
    \item \textbf{Fusion-aware refinement and analysis.}
We show that individual target ranks do not determine fused recommendation utility, motivating feedback on query combinations. A cutoff-risk decomposition characterizes the depth weighting of the log-rank objective.
    \item \textbf{Controlled evidence for the learning design.}
Experiments distinguish the effect of supervision source from subsequent policy refinement and evaluate recommendation quality across three domains.
\end{itemize}

\begin{figure*}[t]
\centering
\includegraphics[width=\textwidth]{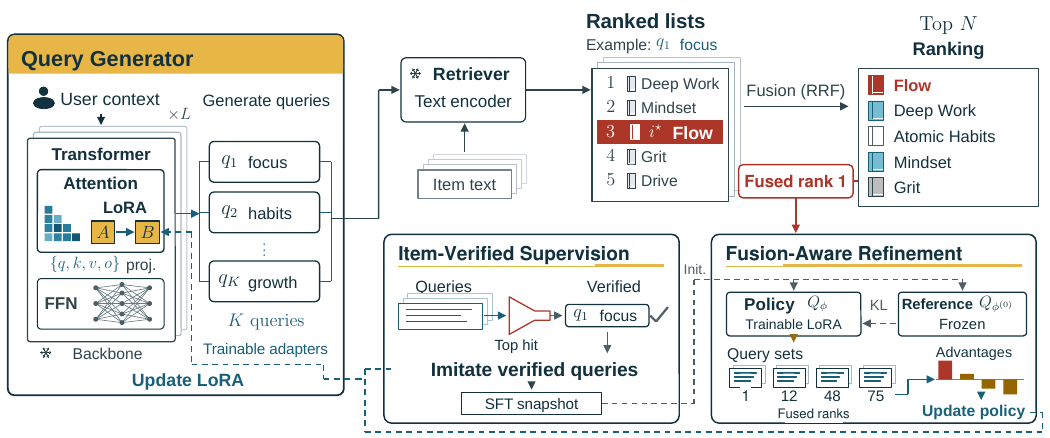}
\caption{\textbf{\name{} overview.} Item-verified supervision constructs SFT targets from queries that recover observed items.
Fusion-aware refinement applies GRPO to target ranks in the RRF-fused lists. The policy and frozen reference start from the SFT checkpoint; the policy is the query generator shown on the left.
KL regularization constrains the policy update. Only LoRA adapters are trained; RRF combines retrieval lists at inference.
}
\label{fig:pipeline}
\Description{User context conditions a transformer query generator with trainable LoRA adapters. A frozen text encoder retrieves item lists, which RRF combines. The item-verified supervision block filters candidate queries by retrieval outcome, imitates a retained query, and creates the SFT snapshot. A branched initialization path connects this snapshot to the trainable policy and frozen reference. The fusion-aware refinement block connects RRF-fused target ranks to group-relative advantages, GRPO policy updates, and the generator's LoRA adapters. The observed item supplies feedback and is not a generator input.}
\end{figure*}

%% file: sections/02_prelim.tex
\section{Preliminaries}
\label{sec:prelim}

Let $\mathcal{U}$ and $\mathcal{I}$ denote the user and item sets.
For user $u$, let $H_u$ denote the chronological interaction history.
The user context $\su$ represents the user's recommendation needs, primarily through their interaction history.
Given $\su$, the system retrieves and ranks items from $\mathcal{I}\setminus H_u$ as recommendations, where the held-out next item $i^\ast$ serves as the target.

Interaction logs provide item targets without annotated queries.
For a chronological sequence $(i_1,\ldots,i_n)$, the final interaction $i_n$ is reserved for evaluation.
An earlier interaction $i_t$ can serve as a training target with prefix history $H_u^{(t)}=(i_1,\ldots,i_{t-1})$.
Let $\mathbf{s}_u^{(t)}$ denote the available context formed from this prefix.
The training examples are drawn from
\begin{equation}
  \mathcal{D}\subseteq
  \bigcup_{u\in\mathcal{U}}
  \bigl\{(\mathbf{s}_u^{(t)},i_t)\mid 2\leq t<n\bigr\}.
  \label{eq:sliding_window}
\end{equation}

%% file: sections/03_framework.tex
\section{\name{}}
\label{sec:framework}

\subsection{Query Generation and Retrieval}
\label{sec:method}

\name{} learns a query generator $\Qphi(q\mid\su)$ from item retrieval outcomes while keeping the retriever $R(i|q)$ fixed (Figure~\ref{fig:pipeline}).
The generator expresses the user context as natural-language retrieval queries.
We denote a generated query set by $q=(q_1,\ldots,q_K)$, where each $q_j$ defines a retrieval request.
At inference, we infer $K$ queries and combine their retrieval results by rank fusion.
The query budget $K$ can be varied at inference without changing the learned parameters.

The frozen retriever $R$ maps each query to a ranked list over $\mathcal{I}\setminus H_u$.
Let $\mathrm{rank}_R(i\mid q_j)$ denote the item's rank in this list.
Queries and item documents are encoded in the same semantic embedding space, and item embeddings can be cached.
While the generator learns to express preferences in free-form text, it retrieves relevant items in this embedding space.
Finally, Reciprocal Rank Fusion (RRF)~\citep{cormack2009rrf} combines the $K$ retrieval lists.
We return the top-$N$ items in the fused ranking as recommendations.
Let $\rho_F(i\mid q)$ denote the rank of item $i$ after applying this fusion to the query set $q$.
We use the same retrieval pool, history masking, fusion settings, and tie handling when evaluating training query sets.

\subsection{Training Objective}
\label{sec:training_objective}

The two stages connect different levels of item feedback.
\emph{Item-verified supervision} selects admissible retrieval expressions for initial training; \emph{fusion-aware refinement} evaluates the recommendation list produced by newly generated query sets.
We evaluate the latter through the observed next item's rank and regularize the generator toward a reference policy trained on verified queries.
Let $\Qphizero$ denote the fixed reference obtained by supervised fine tuning (SFT) on retrieval-verified queries.

\Needspace{8\baselineskip}
For a training pair $(\mathbf{s}^{(t)},i_t)$, the regularized objective is
\begin{align}
  \mathcal{L}(\phi\mid\mathbf{s}^{(t)},i_t)
  &=\underbrace{\mathbb{E}_{q\sim\Qphi(\cdot\mid\mathbf{s}^{(t)})}
    [\ell_R(i_t,q)]}_{\text{retrieval rank loss}}
    \nonumber\\
  &\quad+\beta\underbrace{D_{\mathrm{KL}}\!\left(
    \Qphi(\cdot\mid\mathbf{s}^{(t)})\,\Vert\,
    \Qphizero(\cdot\mid\mathbf{s}^{(t)})\right)}_{\text{reference regularization}},
  \label{eq:training_objective}
\end{align}
Here,
$\ell_R(i_t,q)=\log\!\bigl(1+\rho_F(i_t\mid q)\bigr)$.
This rank loss evaluates the target position in the combined recommendation list produced by the query set.
It requires retrieval ranks and no probability normalization over items.
The KL term penalizes deviation from the SFT reference.

The objective measures recovery of an observed interaction through retrieval, allowing different query phrasings to receive feedback.
We learn the generator in two stages: SFT constructs the initialization and reference policy, followed by group-relative policy updates using retrieval ranks.
The latter use a policy surrogate for discrete generation and ranking, with a log-probability-based estimate of the regularization term.

\subsection{Item-Verified Supervision}
\label{sec:training}

Interaction logs specify which item was selected but do not provide the query that would retrieve it.
We construct supervision by generating candidate queries and verifying them with $R$.
The observed item provides a criterion for evaluating candidate expressions without fixing their wording.
For each training pair $(\mathbf{s}_u^{(t)},i_t)$, the base LLM generates a candidate set $\mathcal{Q}_u^{(t)}$ of $M$ query strings.
We retain queries that retrieve $i_t$ within the acceptance threshold $\tau$:
\begin{equation}
  \hat{\mathcal{Q}}_u^{(t)}
  =\bigl\{q_j\in\mathcal{Q}_u^{(t)}
  \mid\mathrm{rank}_R(i_t\mid q_j)\leq\tau\bigr\}.
  \label{eq:filter}
\end{equation}
Only pairs with at least one accepted query enter the SFT set.
We use the retained queries to construct training responses $v$, yielding a dataset $\hat{\mathcal{D}}=\{(\mathbf{s},v)\}$.
The acceptance threshold $\tau$ selects initial supervision independently of the recommendation cutoff $N$.
The target item is used by the retrieval verifier and is not appended to the generator's input.
Consequently, a query can become a training target through its retrieval behavior even when it differs from the observed item's title.
For a training response $v$ with tokens $v_\ell$, we minimize
\begin{equation}
  \mathcal{L}_{\mathrm{SFT}}(\phi)
  =-\sum_{(\mathbf{s},v)\in\hat{\mathcal{D}}}
  \sum_{\ell=1}^{|v|}
  \log Q_\phi(v_\ell\mid\mathbf{s},v_{<\ell}).
  \label{eq:sft}
\end{equation}
We update LoRA adapters~\citep{hu2022lora} and keep the LLM backbone and retriever frozen.
The resulting checkpoint $\Qphizero$ initializes reinforcement learning and serves as its reference policy.

\subsection{Retrieval Rank Feedback}
\label{sec:objective}

SFT imitates accepted queries, while reinforcement learning compares the retrieval outcomes of newly generated query sets.
For each training context $\mathbf{s}$ and target $i_t$, we generate $G$ candidate query sets $q^{(1)},\ldots,q^{(G)}$ from $\Qphi(\cdot\mid\mathbf{s})$.
Each set contains $K$ queries. We execute the queries, fuse their retrieval lists with RRF, and compute the reward from the resulting rank of $i_t$:
\begin{equation}
  r^{(g)}=-\log\!\left(1+\rho_F(i_t\mid q^{(g)})\right).
  \label{eq:reward}
\end{equation}
The reward evaluates the combined recommendation list through the same fusion operator used to produce recommendations.
Log-rank feedback distinguishes fused target ranks beyond the recommendation cutoff, where cutoff-based hit rewards are tied.\footnote{An earlier Books study recorded approximately $88\%$ zero reward-gradient steps with binary hit feedback, compared with $6.4\%$ with log-rank feedback.}

On the same rollouts, strictly monotonic rank feedback preserves at least as many groups with nonzero centered advantages as binary hit feedback; Appendix~\ref{sec:analysis_resolution} establishes this comparison.

\subsection{Fusion-Aware Refinement}
\label{sec:policy_optimization}
The counts $K$ and $G$ have distinct roles.
$K$ controls how many retrieval requests a single response contains; $G$ controls how many alternative responses are compared for one context.
With the default $K=5$ and $G=4$, we request four responses of five queries and assign one reward to each response.

We use GRPO's within-group reward normalization~\citep{shao2024deepseekmath} to compute the group relative advantage:
\begin{equation}
  \hat{A}^{(g)}
  =\frac{r^{(g)}-\bar{r}}{\sigma_r+\epsilon},
  \label{eq:advantage}
\end{equation}
where $\bar{r}$ and $\sigma_r$ are the mean and standard deviation of the $G$ rewards for the same context, and $\epsilon>0$ stabilizes the denominator.
The group mean provides a context-specific baseline: a positive $\hat{A}^{(g)}$ identifies a response with an above-average retrieval reward, while a negative value identifies one with a below-average reward.
Dividing by the group standard deviation normalizes the scale of these relative comparisons.
Every query in a response shares the same set-level advantage.
If all rewards in a group coincide, the centered advantages are zero.
In particular, $G=1$ removes the relative reward signal.
The policy loss combines the advantage signal with a KL penalty:
\begin{align}
  \mathcal{L}_{\mathrm{RL}}(\phi)
  &=-\frac{1}{G}\sum_{g=1}^{G}
    \hat{A}^{(g)}\frac{\log\Qphi(q^{(g)}\mid\mathbf{s})}{T_g}
    \nonumber\\
  &\quad+\beta D_{\mathrm{KL}}\!\left(
    \Qphi(\cdot\mid\mathbf{s})\,\Vert\,
    \Qphizero(\cdot\mid\mathbf{s})\right).
  \label{eq:grpo}
\end{align}
The first term is the policy learning term, averaged over the $G$ generated responses.
Here $\log Q_\phi(q^{(g)}\mid\mathbf{s})$ is the sum of token log probabilities for the generated response, and $T_g$ is its token count.
Their ratio uses the mean token log probability so that the term does not grow simply by summing over a longer response.
With advantages held fixed during the update, minimizing the negative advantage-weighted term encourages higher probability for above-average responses and lower probability for below-average ones.
The resulting advantage updates the complete response and is determined by its fused target rank.

The second term regularizes the current policy toward the frozen SFT policy $Q_{\phi^{(0)}}$.
The direction $D_{\mathrm{KL}}(Q_\phi\Vert Q_{\phi^{(0)}})$ penalizes assigning probability to outputs that the reference considers unlikely.
The coefficient $\beta$ controls how strongly this reference constrains rank-driven updates: a larger value favors staying closer to the retrieval-verified initialization.
We compute the regularization term from response log probabilities under the current policy and the frozen reference.
We use one gradient update per rollout group with gradient norm clipping, without reusing rollouts across multiple optimization epochs.

\subsection{From Item Feedback to Recommendation Utility}
\label{sec:recommendation_analysis}

The design separates three questions: what an observed item can supervise, how query combinations affect item competition, and which recommendation depths the loss emphasizes.
An item label specifies a retrieval target without prescribing a unique expression.
Item-verified supervision uses it to define admissible queries; Appendix~\ref{sec:analysis_verification} formalizes this selection by conditioning on acceptance.

Admissibility alone does not characterize the outcome of combining queries.
Lemma~\ref{lem:rank_vector} gives two sets of retrieval lists with identical target-rank vectors but different fused target ranks.
The difference arises from competing items across the lists: even the full vector of individual target ranks, and therefore its best or mean, cannot determine the fused outcome in general.
Fusion-aware refinement observes this combined ranking when comparing generated sets.

Having specified the feedback level, the log-rank objective determines its weighting across recommendation depths.
\begin{proposition}[Weighted cutoff-risk decomposition]
\label{prop:weighted_cutoff}
Fix an event with $L\geq1$ eligible items and one target.
Let $Z\in\{1,\ldots,L\}$ be its fused rank under any distribution over generated query sets, and let $h_n=\Pr(Z\leq n)$.
Then
\begin{equation}
  \mathbb{E}[\log(1+Z)]
  =\log 2+\sum_{n=1}^{L-1}w_n(1-h_n),
  \quad w_n=\log\frac{n+2}{n+1}.
  \label{eq:topn_rank_risk}
\end{equation}
\end{proposition}
\begin{proof}[Proof sketch]
Expand $\log(1+Z)-\log2$ into adjacent log increments and take expectations. Appendix~\ref{sec:analysis_topn} gives the full derivation.
\end{proof}
With one target per event, $h_n$ is the expected Recall@$n$.
Thus, the retrieval loss is a depth-weighted recommendation risk: it weights missed-target events across cutoffs, with positive weights that decrease toward the tail.
An equal one-position improvement receives more reward near the head, while improvements beyond a particular cutoff still provide feedback.
The identity explains the objective's weighting across depths; it does not imply improvement at every individual cutoff.

%% file: sections/05_empirical.tex
\section{Empirical Study}
\label{sec:empirical}

Our experiments address the following research questions:
\begin{itemize}[leftmargin=*, noitemsep, topsep=3pt]
\item \textbf{RQ1:} How does \name{} compare with existing recommenders on public datasets?
\item \textbf{RQ2:} How do its training stages and components contribute to recommendation quality?
\item \textbf{RQ3:} How sensitive is \name{} to model choices, hyperparameters, and domain changes?
\item \textbf{RQ4:} How does query retrieval compare with agentic candidate ranking?
\end{itemize}

\subsection{Setup}
\label{sec:setup}
We evaluate on three public datasets Yelp, MovieTV, and Amazon Books~\citep{xu2025iagent}.
We use the benchmark's evaluation splits and train a separate generator for each domain.
Table~\ref{tab:dataset_stats} summarizes the data.

\begin{table}[!ht]
\centering
\small
\caption{Dataset statistics. Avg.\ length is the mean interaction-history length.}
\label{tab:dataset_stats}
\renewcommand{\arraystretch}{1.10}
\setlength{\tabcolsep}{4pt}
\begin{tabular}{@{}lrrrr@{}}
\toprule
\textbf{Dataset} & \textbf{Users} & \textbf{Items} & \textbf{Interactions} & \textbf{Avg.\ length} \\
\midrule
Yelp      & 3.0K  & 31.6K  & 63.1K  & 21.4 \\
MovieTV   & 5.6K  & 29.0K  & 79.7K  & 14.1 \\
Books     & 7.4K  & 120.9K & 207.8K & 28.2 \\
\bottomrule
\end{tabular}
\end{table}

The main retrieval evaluation ranks the full item set without negative sampling and includes all 2,950, 5,649, and 7,377 held out users on Yelp, MovieTV, and Books, respectively.
Unless stated otherwise, experiments use these evaluation sets; the group-size, KL-weight, and inference query-count scans use the same fixed subset of 1,000 users per domain.
History items are excluded from retrieval.
We select the SFT acceptance threshold $\tau$, $G$, $\beta$, and inference query count $K$ on validation data; the remaining \name{} hyperparameters use fixed defaults.
No configuration is selected or adjusted using test results.
All text retrieval methods index the same documents, formed from item titles and cleaned descriptions.
For Yelp, descriptions include categories, city, and star rating.
For MovieTV, they include the synopsis after removing physical media and condition notes.

\begin{table*}[!t]
\centering
\footnotesize
\caption{Main recommendation results (\%). \textbf{Bold}/\underline{underline}: best/second best.}
\label{tab:main_results}
\resizebox{\textwidth}{!}{%
\setlength{\tabcolsep}{1.2pt}
\begin{tabular}{l cccc cccc  cccc cccc  cccc cccc}
\toprule
 & \multicolumn{8}{c}{\textbf{Yelp}} & \multicolumn{8}{c}{\textbf{MovieTV}} & \multicolumn{8}{c}{\textbf{Books}} \\
\cmidrule(lr){2-9}\cmidrule(lr){10-17}\cmidrule(lr){18-25}
 & \multicolumn{4}{c}{Recall@$N$} & \multicolumn{4}{c}{NDCG@$N$} & \multicolumn{4}{c}{Recall@$N$} & \multicolumn{4}{c}{NDCG@$N$} & \multicolumn{4}{c}{Recall@$N$} & \multicolumn{4}{c}{NDCG@$N$} \\
Method & @5 & @10 & @20 & @50 & @5 & \textbf{@10} & @20 & @50 & @5 & @10 & @20 & @50 & @5 & \textbf{@10} & @20 & @50 & @5 & @10 & @20 & @50 & @5 & \textbf{@10} & @20 & @50 \\
\midrule
SASRec & 0.17 & 0.27 & 0.68 & 1.46 & 0.10 & 0.13 & 0.23 & 0.39 & 4.87 & 6.76 & 9.17 & 13.47 & 3.70 & 4.30 & 4.91 & \underline{5.76} & 1.74 & 1.91 & 2.10 & 2.60 & 1.52 & 1.57 & 1.62 & 1.72 \\
BERT4Rec & 0.44 & 0.64 & 1.08 & 2.14 & 0.31 & 0.37 & 0.48 & 0.69 & 4.71 & 6.57 & \underline{9.22} & \textbf{13.95} & 3.43 & 4.03 & 4.70 & 5.63 & 1.34 & 1.60 & 1.88 & 2.55 & 1.15 & 1.23 & 1.30 & 1.44 \\
LightGCN & 0.61 & 1.29 & 2.20 & 4.17 & 0.33 & 0.55 & 0.78 & 1.16 & 2.69 & 4.32 & 6.36 & 10.78 & 1.84 & 2.35 & 2.87 & 3.74 & 1.87 & 2.48 & 3.10 & 4.15 & 1.40 & 1.60 & 1.75 & 1.96 \\
UniSRec & 0.98 & 1.46 & 2.17 & 5.08 & 0.58 & 0.73 & 0.91 & 1.47 & 4.73 & 6.37 & 8.66 & 13.05 & 3.46 & 4.00 & 4.57 & 5.43 & 2.72 & 3.52 & 4.36 & 5.95 & 2.15 & 2.42 & 2.63 & 2.94 \\
TIGER & 0.51 & 0.78 & 1.25 & 1.25 & 0.35 & 0.44 & 0.56 & 0.56 & 2.76 & 3.88 & 5.51 & 5.51 & 2.05 & 2.41 & 2.82 & 2.82 & 0.84 & 0.87 & 0.95 & 0.95 & 0.77 & 0.77 & 0.79 & 0.79 \\
LETTER & 0.27 & 0.68 & 1.08 & 1.08 & 0.19 & 0.32 & 0.42 & 0.42 & 2.83 & 3.88 & 5.63 & 5.63 & 2.21 & 2.54 & 2.99 & 2.99 & 0.87 & 0.95 & 1.07 & 1.07 & 0.76 & 0.79 & 0.82 & 0.82 \\
\midrule
BM25 & 0.85 & 1.25 & 1.66 & 2.64 & 0.56 & 0.68 & 0.79 & 0.98 & 4.39 & 6.18 & 8.13 & 10.83 & 3.11 & 3.69 & 4.18 & 4.73 & 3.82 & 4.87 & 5.99 & 7.75 & 2.80 & 3.14 & 3.43 & 3.78 \\
InputRetrieval & 0.54 & 0.81 & 1.32 & 3.39 & 0.43 & 0.51 & 0.63 & 1.04 & 2.21 & 3.38 & 4.69 & 7.70 & 1.44 & 1.81 & 2.14 & 2.74 & 2.97 & 4.27 & 6.14 & 10.63 & 1.98 & 2.40 & 2.87 & 3.75 \\
BIGRec & 1.59 & 2.17 & 3.22 & 5.15 & 1.20 & 1.38 & 1.65 & 2.03 & \underline{5.81} & \underline{6.82} & 7.61 & 9.10 & \underline{4.58} & \underline{4.90} & \underline{5.10} & 5.40 & \underline{6.40} & 7.62 & 8.76 & 10.34 & \underline{4.99} & \underline{5.39} & \underline{5.67} & 5.99 \\
GPT4Rec & \underline{2.10} & \underline{2.92} & \underline{3.93} & \underline{6.58} & \underline{1.37} & \underline{1.63} & \underline{1.88} & \underline{2.41} & 5.08 & 6.43 & 7.79 & 10.02 & 3.57 & 4.00 & 4.35 & 4.78 & 6.14 & 7.39 & 8.74 & 10.74 & 4.76 & 5.16 & 5.50 & 5.89 \\
QueRec & 1.36 & 2.58 & 3.86 & 5.97 & 0.91 & 1.30 & 1.62 & 2.04 & 3.98 & 4.94 & 6.21 & 8.20 & 3.09 & 3.40 & 3.72 & 4.12 & 3.32 & 4.45 & 6.03 & 7.94 & 2.37 & 2.74 & 3.13 & 3.51 \\
QueryAgent-R1 & 1.08 & 1.90 & 3.05 & 4.88 & 0.70 & 0.95 & 1.25 & 1.61 & 2.94 & 3.81 & 5.01 & 6.71 & 2.10 & 2.38 & 2.68 & 3.02 & 6.07 & \underline{7.66} & \underline{9.37} & \underline{11.71} & 4.63 & 5.15 & 5.58 & \underline{6.05} \\
\midrule
\rowcolor{blue!8}
\textbf{\name{}} & \textbf{3.22} & \textbf{4.75} & \textbf{6.34} & \textbf{10.24} & \textbf{2.26} & \textbf{2.75} & \textbf{3.15} & \textbf{3.91} & \textbf{6.50} & \textbf{8.00} & \textbf{10.25} & \underline{13.67} & \textbf{4.59} & \textbf{5.07} & \textbf{5.64} & \textbf{6.32} & \textbf{9.00} & \textbf{11.32} & \textbf{13.56} & \textbf{16.38} & \textbf{6.91} & \textbf{7.66} & \textbf{8.23} & \textbf{8.79} \\
\bottomrule
\end{tabular}
}
\end{table*}

We use \texttt{Qwen2.5-\allowbreak 7B-\allowbreak Instruct}~\citep{qwen25} as the autoregressive query generator.
History items enter the prompt as titles, with the most recent last.
Descriptions and metadata are available to the retriever but are excluded from the generator's history input.
A shared prompt template fills domain specific slots for item type, vocabulary hints, and query examples.
The prompt templates remain fixed within each comparison.
Dense retrieval methods use the frozen \texttt{all-\allowbreak MiniLM-\allowbreak L6-\allowbreak v2} encoder~\citep{reimers2019sbert}.
We use dot products of normalized query and item embeddings.

For open-pool recommendation, we report Recall@$N$ and NDCG@$N$ for $N\in\{5,10,20,50\}$, with NDCG@10 as the primary metric.
Each user has one relevant target, so the ideal DCG is 1.
Unless stated otherwise, \name{} generates $K=5$ queries.

\subsection{Implementation Details}
\label{sec:impl}

We implement \name{} in PyTorch and train on a single NVIDIA A800 GPU with 80 GB memory.
We search the SFT acceptance threshold $\tau$ from 50 to 1,000, using the same training-only candidate pool and SFT settings for each threshold.
Each domain's threshold is selected by validation NDCG@10 after SFT and fixed before policy refinement.

Both training stages use LoRA~\citep{hu2022lora} on the attention query, key, value, and output projections $(W_Q,W_K,W_V,W_O)$.
We set the adapter rank to $r=16$ and the scaling parameter to $\alpha=32$.
The backbone and retrieval encoder remain frozen.
SFT runs for three epochs with learning rate $2\times10^{-4}$, per-device batch size 2, gradient accumulation over eight batches, and a 3\% warmup ratio.
The SFT checkpoint remains fixed as the KL reference.

We search $G$ from 1 to 32, $\beta$ from 0 to 1, and inference query count $K$ from 1 to 20.
For all baseline models, we carefully tune hyperparameters on the validation sets and report test performance using the best validation-selected configuration.
Results are averaged over five random runs. Reported improvements over the best-performing baselines are statistically significant ($p<0.05$), assessed using a two-sided $t$-test.

\subsection{Baselines}
\label{sec:baselines}

We compare two groups of baselines.
\textbf{\textit{Discriminative and generative recommendation.}}
\textbf{SASRec}~\citep{kang2018sasrec} and \textbf{BERT4Rec}~\citep{sun2019bert4rec} model interaction sequences; \textbf{LightGCN}~\citep{he2020lightgcn} learns from the interaction graph; \textbf{UniSRec}~\citep{hou2022unisrec} incorporates item text for transferable sequential recommendation.
\textbf{TIGER}~\citep{rajput2023tiger} and \textbf{LETTER}~\citep{letter2024} generate item identifiers.
\textbf{\textit{Query-based recommendation.}}
\textbf{BM25}~\citep{robertson2009bm25} and \textbf{InputRetrieval} perform lexical and dense retrieval, respectively, from the same user context and item documents as \name{}, without an LLM.
\textbf{BIGRec}~\citep{bao2023bigrec} generates target-item text and grounds it to the catalog. \textbf{GPT4Rec}~\citep{li2023gpt4rec} trains a query generator using item text; the concurrent \textbf{QueryAgent-R1}~\citep{queryagent2026} uses an agentic search procedure to refine e-commerce query recommendations.
\textbf{QueRec}~\citep{han2025querec} combines personalized-query retrieval with collaborative scores through adaptive reranking, without LLM fine tuning.
Trainable text-generation baselines use their respective training objectives.
GPT4Rec and BIGRec use the same Qwen2.5-7B backbone, frozen MiniLM retriever, $K=5$ query budget, and rank-fusion procedure as \name{}.
The zero-shot generator uses the same backbone, prompt, and retrieval pipeline as \name{}, without task training.

\subsection{Results on Public Datasets (RQ1)}
\label{sec:perdomain-results}
Table~\ref{tab:main_results} reports recommendation results over the full item set. We make four observations.

\begin{itemize}[leftmargin=*, noitemsep, topsep=3pt]
\item \textbf{Leading results across domains.}
\name{} achieves the best results among the evaluated methods in 23 of the 24 dataset--metric comparisons. Its relative NDCG@10 gains over the strongest baseline are $68.7\%$, $3.5\%$, and $42.1\%$ on Yelp, MovieTV, and Books, respectively.

\item \textbf{Strong performance on short recommendation lists.}
On MovieTV, \name{} improves Recall@5 and Recall@10 by $11.9\%$ and $17.3\%$ over the strongest baseline. BERT4Rec leads at Recall@50 ($13.95\%$ versus $13.67\%$), where \name{} ranks second, while \name{} leads at every NDCG cutoff. This pattern highlights its effectiveness at placing relevant items near the top of the recommendation list.

\item \textbf{Effectiveness of generated retrieval queries.}
Interestingly, LLM-based query recommenders provide the strongest baseline in 21 of the 24 comparisons. GPT4Rec leads the baselines on Yelp, while BIGRec and QueryAgent-R1 provide the strongest results on Books. These results support generated queries as an effective interface between user context and item retrieval.

\item \textbf{Gains under a shared inference pipeline.}
With the same backbone, retriever, query budget, and rank-fusion procedure as GPT4Rec and BIGRec, \name{} outperforms both on every reported metric across all three domains. The advantage therefore persists when these inference components are held fixed. The following ablations examine the contributions of retrieval verification and policy refinement.
\end{itemize}

\subsection{Ablation Study (RQ2)}
\label{sec:ablations}

\begin{table}[!b]
\centering
\small
\caption{MovieTV component ablations (\%). \textbf{Bold}: best value per metric. The w/o multi-query variant uses $K=1$ at inference with the Full checkpoint.}
\label{tab:grpo_ablation}
\resizebox{\columnwidth}{!}{%
\setlength{\tabcolsep}{3.4pt}
\begin{tabular}{l cccc cccc}
\toprule
 & \multicolumn{4}{c}{\textbf{Recall@$N$}} & \multicolumn{4}{c}{\textbf{NDCG@$N$}} \\
Method & @5 & @10 & @20 & @50 & @5 & \textbf{@10} & @20 & @50 \\
\midrule
Zero-shot & 4.07 & 5.45 & 7.13 & 9.42 & 2.93 & 3.37 & 3.79 & 4.24 \\
SFT only & 5.08 & 7.06 & 8.92 & 11.44 & 3.69 & 4.33 & 4.80 & 5.30 \\
\textbf{Full} & \textbf{6.50} & \textbf{8.00} & \textbf{10.25} & \textbf{13.67} & \textbf{4.59} & \textbf{5.07} & \textbf{5.64} & \textbf{6.32} \\
\midrule
w/o verification & 4.76 & 6.36 & 7.98 & 11.13 & 3.52 & 4.03 & 4.44 & 5.07 \\
w/o KL ($\beta{=}0$) & 5.58 & 7.43 & 9.10 & 12.13 & 4.06 & 4.67 & 5.09 & 5.69 \\
$G{=}1$ & 5.26 & 7.10 & 8.87 & 11.52 & 3.77 & 4.37 & 4.82 & 5.34 \\
w/o multi-query & 5.93 & 7.45 & 9.17 & 11.63 & 4.26 & 4.75 & 5.19 & 5.68 \\
\bottomrule
\end{tabular}%
}
\end{table}

For the component ablations, we report mean performance over five random runs and assess statistical significance at the $p<0.05$ level.
Table~\ref{tab:grpo_ablation} examines individual components on MovieTV using the same generator backbone and retriever.
Full achieves statistically significant NDCG@10 gains over the component ablations.
SFT only achieves $4.33\%$ NDCG@10, compared with $5.07\%$ for Full. Policy refinement adds $0.74$ percentage points over SFT only.

\begin{itemize}[leftmargin=*, noitemsep, topsep=3pt]
\item \textbf{Training and retrieval verification.}
Removing retrieval filtering from SFT data construction, while retaining subsequent GRPO, reduces MovieTV NDCG@10 from $5.07\%$ to $4.03\%$, a drop of $1.04$ percentage points. This controlled comparison supports filtering the generated supervision before training.

\item \textbf{Group feedback and KL.}
On MovieTV, $G=1$ yields $4.37\%$ NDCG@10, compared with $5.07\%$ for Full. At $G=1$, Eq.~\ref{eq:advantage} gives zero reward advantage, leaving only regularization.
The w/o KL variant reaches $4.67\%$ NDCG@10, compared with $5.07\%$ for Full.

\item \textbf{Multiple retrieval queries.}
Using $K=1$ at inference with the Full checkpoint lowers NDCG@10 from $5.07\%$ to $4.75\%$ and Recall@10 from $8.00\%$ to $7.45\%$. The checkpoint is held fixed, so this comparison measures the effect of the inference query count.
\end{itemize}

\begin{table}[t]
\fontsize{7.5}{8.8}\selectfont
\centering
\caption{Retrieval verification across domains (\%). The w/o verification variant retains all generated pseudo-labels; subsequent training stages are unchanged.}
\label{tab:filtering_all_domains}
\setlength{\tabcolsep}{1.6pt}
\begin{tabular}{@{}l@{\hspace{4pt}}lrrrrrrrr@{}}
\toprule
 & & \multicolumn{4}{c}{\textbf{Recall@$N$}} & \multicolumn{4}{c}{\textbf{NDCG@$N$}} \\
Domain & Method & @5 & @10 & @20 & @50 & @5 & @10 & @20 & @50 \\
\midrule
\multirow{2}{*}{Yelp} & Full & 3.22 & 4.75 & 6.34 & 10.24 & 2.26 & 2.75 & 3.15 & 3.91 \\
 & w/o verification & 2.88 & 4.03 & 6.14 & 9.76 & 1.86 & 2.22 & 2.75 & 3.47 \\
\midrule
\multirow{2}{*}{MovieTV} & Full & 6.50 & 8.00 & 10.25 & 13.67 & 4.59 & 5.07 & 5.64 & 6.32 \\
 & w/o verification & 4.76 & 6.36 & 7.98 & 11.13 & 3.52 & 4.03 & 4.44 & 5.07 \\
\midrule
\multirow{2}{*}{Books} & Full & 9.00 & 11.32 & 13.56 & 16.38 & 6.91 & 7.66 & 8.23 & 8.79 \\
 & w/o verification & 8.70 & 10.90 & 13.24 & 16.08 & 6.64 & 7.35 & 7.95 & 8.51 \\
\bottomrule
\end{tabular}
\end{table}

\begin{table}[t]
\centering
\small
\caption{Supervision source on MovieTV: NDCG@10 (\%). $\Delta$ is the gain from subsequent GRPO over the corresponding SFT initialization.}
\label{tab:supervision_source}
\setlength{\tabcolsep}{6pt}
\renewcommand{\arraystretch}{1.08}
\begin{tabular}{@{}lrrr@{}}
\toprule
Supervision & SFT only & + GRPO & $\Delta$ \\
\midrule
Item titles & 3.14 & 3.71 & +0.57 \\
Retrieval-verified queries & 4.33 & \textbf{5.07} & +0.74 \\
\bottomrule
\end{tabular}
\end{table}

Table~\ref{tab:supervision_source} traces both supervision sources through training with the same backbone and inference pipeline.
Title SFT improves from $3.14\%$ to $3.71\%$ after GRPO, while retrieval-verified SFT improves from $4.33\%$ to $5.07\%$.
Policy refinement benefits both initializations, and retrieval-verified supervision provides the stronger starting point and final result.
The $1.36$-percentage-point final gap shows the value of the supervision source even when both pipelines include policy refinement.

Additional verification results appear in Table~\ref{tab:filtering_all_domains}.

\subsection{Hyperparameter Sensitivity (RQ3)}
\label{sec:hparam}

Figure~\ref{fig:hparam} examines group size $G$, KL weight $\beta$, and query count $K$ on a fixed 1,000-user subset of each domain, using NDCG@10 and Recall@10.
These scans describe sensitivity; configuration selection uses the validation protocol in Section~\ref{sec:setup}.
Group size has no common optimum: MovieTV attains its highest NDCG@10 at $G=32$, whereas Books peaks at $G=8$. Performance varies non-monotonically across the displayed KL weights. Query-count effects differ across domains: MovieTV increases from $4.05\%$ at $K=1$ to $5.11\%$ at $K=20$, with intermediate fluctuations, while Books NDCG@10 varies within $6.75$--$7.01\%$ without a consistent gain as $K$ grows.

\begin{figure}[tbp]
\centering
\includegraphics[width=\columnwidth]{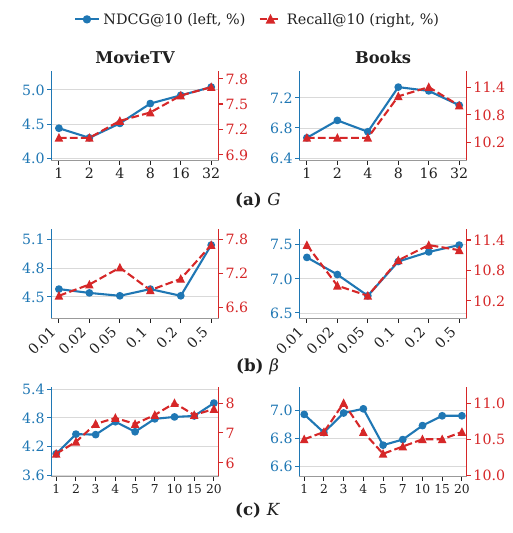}
\caption{Hyperparameter sensitivity on 1,000 sampled users per domain. Blue/left: NDCG@10; red/right: Recall@10. (a) $G$. (b) $\beta$. (c) $K$.}
\label{fig:hparam}
\Description{Six panels are arranged as three parameter rows and MovieTV/Books columns. Both metrics are shown for G and query-count settings and KL weights from 0.01 through 0.5.}
\end{figure}

\subsection{Generator and Retriever Analysis (RQ3)}
\label{sec:backbone}
\label{para:backbone}

\subsubsection{Query generator families}
Figure~\ref{fig:generator_training} compares zero-shot and Full results for three query generators across the three domains.
On Yelp, Llama-3-8B~\citep{llama3} achieves $3.01\%$ NDCG@10 and outperforms Qwen ($2.75\%$) and Mistral-7B~\citep{mistral7b} ($2.51\%$) on this metric.
Qwen achieves the highest NDCG@10 on MovieTV and Books, at $5.07\%$ and $7.66\%$.
Llama reaches $3.63\%$ and $6.83\%$, while Mistral reaches $2.68\%$ and $6.75\%$ on these domains.
The pipeline is applicable to all three backbones, but the best generator depends on the domain.

\begin{figure}[t]
\centering
\includegraphics[width=\columnwidth]{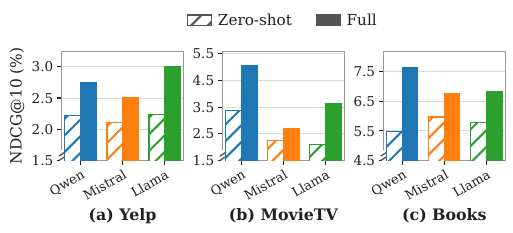}
\caption{Query generator families: adjacent hatched/solid bars compare zero-shot/Full for Qwen2.5-7B, Mistral-7B, and Llama-3-8B.}
\label{fig:generator_training}
\Description{Three domain panels compare generators in the order Yelp, MovieTV, and Books. Qwen is blue, Mistral orange, and Llama green in every panel. Each generator has adjacent hatched zero-shot and solid Full bars. Labeled truncated y axes and break marks make differences visible.}
\end{figure}

\subsubsection{Generator size and retrieval encoder}
Figure~\ref{fig:backbone} separately examines these choices on Yelp and MovieTV \emph{without task training}.
On MovieTV, increasing Qwen from 1.5B to 7B improves NDCG@10 from $2.04\%$ to $3.37\%$, while 14B does not improve further under this fixed prompt.
On MovieTV, bge reaches $3.53\%$, compared with $3.37\%$ for MiniLM and $3.24\%$ for mpnet. MiniLM is close to bge on Yelp.
We retain MiniLM as the common retriever for the main comparisons.

\begin{figure}[tbp]
\centering
\includegraphics[width=\columnwidth]{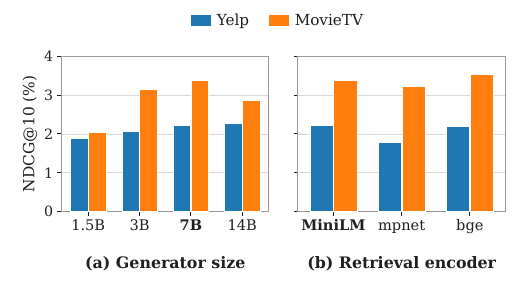}
\caption{Generator and encoder sensitivity without task training. (a) Generator size, with MiniLM. (b) Encoder, with the 7B generator. Bold labels mark defaults.}
\label{fig:backbone}
\Description{Two grouped bar charts compare generator sizes and retrieval encoders on Yelp and MovieTV. Yelp is blue and MovieTV orange; both panels share a zero-based NDCG scale. Bold labels mark the default seven billion parameter generator and MiniLM encoder.}
\end{figure}

\subsection{Cross-Domain Transfer (RQ3)}
\label{sec:generalization}

We train the generator on one domain and evaluate it on another without updating its parameters, retrieving from the target domain's item index.
Figure~\ref{fig:xdomain} organizes the source--target comparisons.
Transfer is asymmetric. In-domain training remains strongest on Yelp and MovieTV, at $2.75\%$ and $5.07\%$. On Books, generators trained on Yelp and MovieTV achieve NDCG@10 of $8.01\%$ and $8.09\%$, respectively, exceeding the Books-trained generator ($7.66\%$). These results indicate transferable query-generation behavior, with performance that depends on the source--target pair.

\begin{figure}[t]
\centering
\includegraphics[width=\columnwidth]{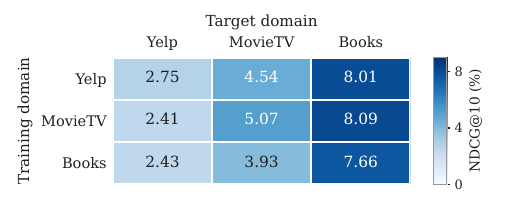}
\caption{Cross-domain transfer: NDCG@10 (\%). Rows identify training domains and columns target domains.}
\label{fig:xdomain}
\Description{The three-by-three transfer matrix is complete. For target Yelp, source Yelp/MovieTV/Books scores are 2.75/2.41/2.43; for target MovieTV they are 4.54/5.07/3.93; for target Books they are 8.01/8.09/7.66.}
\end{figure}

\subsection{Case Study: From History to Queries}
\label{sec:casestudy}

Figure~\ref{fig:case_study} illustrates how a reading history can be expressed as several retrieval queries.
The books suggest related interests in cognitive science, behavioral change, and human achievement.
The example queries express these interests from different angles, with \emph{Flow: The Psychology of Optimal Experience} illustrating a relevant item.
The intermediate language representation makes the connection between historical interests and retrieval requests explicit.

\input{figures/case_books_flow}

\subsection{Agentic Recommendation (RQ4)}
\label{sec:closedpool_eval}
\label{sec:closedpool}

We compare \name{} with i$^2$Agent and MemRec~\citep{memrec2026} in a candidate-ranking setting.
All methods use Qwen2.5-7B and the same candidate set for each user: the target and nine randomly sampled negatives, excluding history items.

The agentic baselines reason directly over candidate content.
\name{} generates queries from user context, then scores the candidates by retrieval; candidate items are not included in its generator prompt.
We summarize this setting with NDCG@5.

\begin{figure}[t]
\centering
\papergraphic{\columnwidth}{figures/fig_closedpool.pdf}
\caption{Candidate-ranking NDCG@5 with the same backbone and 10 candidates per user.}
\label{fig:closedpool}
\Description{A grouped bar chart compares NDCG at five for MemRec, i-squared Agent, and \name{} on Yelp, MovieTV, and Books. MemRec leads on all three domains. Query retrieval exceeds i-squared Agent on Yelp and Books and scores lower on MovieTV.}
\end{figure}

Figure~\ref{fig:closedpool} shows that \name{} reaches NDCG@5 of $60.5\%$, $53.9\%$, and $60.8\%$ on Yelp, MovieTV, and Books.
i$^2$Agent achieves $52.6\%$, $56.5\%$, and $54.3\%$, respectively, while MemRec achieves $67.6\%$, $67.7\%$, and $62.9\%$.
MemRec leads on all three datasets. \name{} outperforms i$^2$Agent on Yelp and Books, while scoring lower on MovieTV.
These results distinguish open-catalog query retrieval from reasoning directly over a supplied candidate list.

%% file: figures/case_books_flow.tex
\begin{figure}[t]
\centering
\resizebox{\columnwidth}{!}{%
\begin{tikzpicture}[x=1cm,y=1cm,>=Latex,
 every node/.style={font=\footnotesize,inner sep=1pt},
 flow/.style={->,draw=black!65,line width=.5pt}]
\path[use as bounding box] (0,0) rectangle (8.4,5.55);
\node[anchor=west,font=\footnotesize\bfseries] at (0,5.35) {User history};
\node[anchor=east,text=black!60,font=\scriptsize] at (8.4,5.35) {Books};
\node[anchor=south west,inner sep=0pt] at (.08,4.20)
 {\includegraphics[height=.91cm]{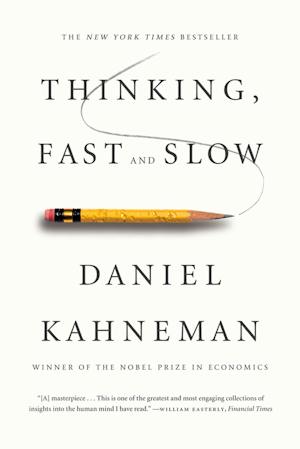}};
\node[anchor=west,text width=1.86cm,align=left] at (.83,4.65) {Thinking, Fast\\and Slow};
\node[anchor=south west,inner sep=0pt] at (3.00,4.20)
 {\includegraphics[height=.91cm]{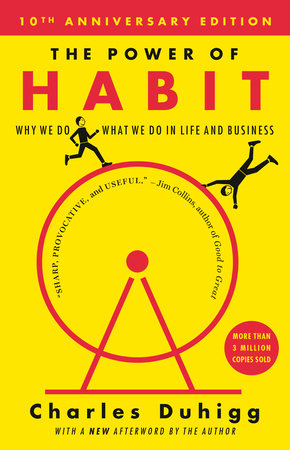}};
\node[anchor=west,text width=1.90cm,align=left] at (3.73,4.65) {The Power\\of Habit};
\node[anchor=south west,inner sep=0pt] at (5.91,4.20)
 {\includegraphics[height=.91cm]{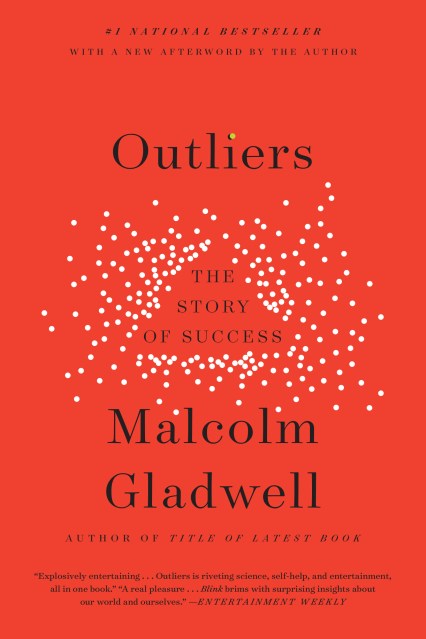}};
\node[anchor=west,text width=1.50cm] at (6.67,4.65) {Outliers};
\draw[black!25,line width=.4pt] (.08,4.07)--(8.32,4.07);
\draw[flow] (4.20,4.04)--(4.20,3.76);
\node[anchor=west,font=\footnotesize\bfseries] at (.08,3.58) {Example queries};
\draw[caseblue,line width=1pt] (.18,1.69)--(.18,3.35);
\node[anchor=west] at (.42,3.20) {\caseaccent{psychology} of peak performance};
\node[anchor=west] at (.42,2.87) {cognitive science and productivity};
\node[anchor=west] at (.42,2.54) {science of mental \caseaccent{flow} states};
\node[anchor=west] at (.42,2.21) {behavioral \caseaccent{psychology} intrinsic motivation};
\node[anchor=west] at (.42,1.88) {popular science human achievement};
\draw[flow] (4.20,1.61)--(4.20,1.24);
\node[anchor=west,font=\scriptsize,text=black!65] at (4.35,1.43) {Retrieve};
\draw[caseblue,line width=.7pt] (.05,.08) rectangle (8.35,1.18);
\node[anchor=south west,inner sep=0pt] at (.17,.14)
 {\includegraphics[height=.98cm]{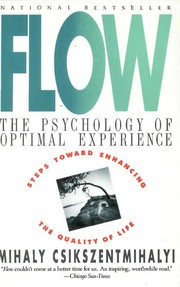}};
\node[anchor=west,text width=5.45cm,align=left] at (1.06,.63)
 {\textbf{\caseaccent{Flow}}: The \caseaccent{Psychology}\\of Optimal Experience};
\end{tikzpicture}}
\caption{Books case study. Five queries express related interests in the history, with \emph{Flow} as an example relevant item. Cover thumbnails identify book titles.}
\label{fig:case_study}
\Description{The history contains Thinking, Fast and Slow, The Power of Habit, and Outliers. Five example queries express interests in psychology, productivity, flow states, motivation, and human achievement. Flow is shown as an example relevant item; no empirical rank is reported.}
\end{figure}

%% file: sections/06_related.tex
\section{Related Work}
\label{sec:related}

\subsection{Query Learning for Recommendation}
Query-based recommendation connects language generation with item retrieval, with the closest approaches differing in how they train or construct the generated text.
GPT4Rec fine-tunes a language model to predict target item titles, then uses the generated text with BM25 to retrieve recommendations~\citep{li2023gpt4rec}.
Its generator learns target-title likelihood; retrieval ranks do not supervise its parameter updates.
BIGRec similarly uses target item text for supervised adaptation before grounding generated text to actual items~\citep{bao2023bigrec}.

Rather than adapting the generator to target text, QueRec~\citep{han2025querec} generates enriched item descriptions and personalized queries through prompting without LLM fine tuning, and integrates semantic retrieval with collaborative signals.
This approach uses an LLM to express preferences, but does not learn its query generation policy from the retrieval outcomes.
Rec-R1 couples LLM-generated text with fixed recommendation models~\citep{lin2025recr1}.
Concurrent QueryAgent-R1 combines memory-based user profiling with a search agent for proactive e-commerce query recommendation~\citep{queryagent2026}.
The focus of \name{} is how item-level observations supervise an intermediate query interface.
Retrieval verification converts those observations into query targets without prescribing their wording, while policy refinement evaluates generated sets at the fused ranking used for recommendation.
This connects supervision construction and recommendation-level outcome optimization through a fixed retrieval system.

\vspace{-1mm}

\subsection{Other Recommendation Approaches}
Discriminative recommenders score items through learned representations, with sequential models deriving user representations from interaction histories.
GRU4Rec uses recurrent networks~\citep{hidasi2015gru4rec}, SASRec applies causal self attention~\citep{kang2018sasrec}, and BERT4Rec uses masked item prediction~\citep{sun2019bert4rec}.
LightGCN captures collaborative structure through graph propagation~\citep{he2020lightgcn}, while UniSRec uses item text to support transfer across domains~\citep{hou2022unisrec}.
RecFormer extends this approach to textual item sequences~\citep{li2023recformer}.
Generative recommenders instead predict item identifiers or text.
P5 casts recommendation tasks as text generation~\citep{geng2022p5}; TIGER generates semantic IDs obtained through residual quantization~\citep{rajput2023tiger}; and OneRec develops generative recommendation at industrial scale~\citep{onerec2024}.
LETTER learns item tokenization with semantic and collaborative regularization~\citep{letter2024}.
\name{} generates retrieval requests that can each match multiple items.

Tags provide another semantic connection between users and items~\citep{xu2011semrec,yu2025tagcf,wu2026guider}.
SemRec combines structural and semantic relationships in social tagging data for item recommendation~\citep{xu2011semrec}.
\name{} instead learns to compose language requests, with observed-item retrieval outcomes supervising their wording.

Natural-language interfaces broaden how users interact with recommenders~\citep{zhang2024instructrec}.
Agentic approaches further introduce memory, tool use, and interaction between agents~\citep{memrec2026,xu2025iagent,zhang2024agentcf, tao2026personalizing, kang2026personalized}.
The candidate-ranking protocols considered here provide a small item set to the LLM, allowing direct reasoning about candidate content.
\name{} instead generates retrieval queries without candidate enumeration in its prompt, and ranks the supplied candidates by retrieval when evaluated in this setting.

%% file: sections/07_discussion.tex
\vspace{-1mm}
\section{Conclusion}
\label{sec:conclusion}

We presented \name{}, a framework for learning recommendation queries from item interactions through a fixed retrieval system.
Its design separates two learning decisions: retrieval verification identifies useful expressions to imitate, and policy refinement evaluates newly generated query sets through their combined item ranking.
The same retrieval interface supports both stages, while only the query generator is updated.
Controlled comparisons show the benefit of retrieval-verified supervision over title imitation and the further gain from policy refinement.
The resulting query interface supports recommendation over the full item set without requiring query annotations or a learned reward model.

Future work can adapt the query budget to user context and examine alternative fusion functions within query learning.
Another direction is to study query refresh policies for changing user interests and item collections.
The separation between query generation and retrieval also supports exploring additional retrieval backends.
Broader questions of verifier reliability when LLMs arbitrate between prior knowledge and retrieved evidence are also relevant to retrieval-augmented systems more generally~\citep{sun2026diagnosing}.

%% file: sections/08_appendix.tex
\section{Inference Latency and Cost}
\label{sec:appendix_latency}

\subsection{Generation Cost}
\label{sec:appendix_generation_cost}

We report generation time separately from retrieval and fusion.
The recorded generation wall times below were measured on a single A800-80GB, with $K=5$ and batch size 1.

\begin{center}
\small
\setlength{\tabcolsep}{3pt}
\begin{tabular}{@{}lrrr@{}}
\toprule
Setting & \#users & Wall time & s/user \\
\midrule
Qwen2.5-7B, MovieTV  & 5{,}649 & 3{,}963\,s & 0.702 \\
Qwen2.5-7B, Books    & 7{,}377 & 5{,}606\,s & 0.760 \\
Qwen2.5-14B, MovieTV & 5{,}649 & 3{,}677\,s & 0.651 \\
Qwen2.5-14B, Books   & 7{,}377 & 6{,}283\,s & 0.852 \\
\bottomrule
\end{tabular}
\end{center}

The generator does not enumerate the item set in its prompt.
Its runtime depends on the model, input and output lengths, and generation settings.

\subsection{Retrieval and Indexing Cost}
\label{sec:appendix_retrieval_cost}

We measure retrieval and fusion on CPU with NumPy against precomputed item embeddings ($d=384$), as summarized below.
\begin{center}
\small
\setlength{\tabcolsep}{7pt}
\begin{tabular}{@{}lrr@{}}
\toprule
Domain & Catalog size & Retrieval + fusion (ms/user) \\
\midrule
Yelp & $32$K & 5.0 \\
MovieTV & $29$K & 5.0 \\
Books & $121$K & 26.6 \\
\bottomrule
\end{tabular}
\end{center}
These measurements use exact retrieval.
GPU execution and approximate indexing offer additional implementation choices as the item set grows.

Exact dense retrieval scores all eligible items, so its work grows with the item set.
Item embeddings are precomputed and reused across requests.
A new domain requires constructing its item index, but our transfer experiments do not update the generator parameters.
We distinguish generation, retrieval, and offline indexing costs.

\Needspace{30\baselineskip}
\section{Pseudocode}
\label{sec:appendix_pseudocode}

Algorithms~\ref{alg:train} and~\ref{alg:infer} summarize training and inference, respectively.
Training generates query sets as responses; inference constructs a set from individually generated queries.

\begin{algorithm}[h]
\caption{\name{} training (one domain)}
\label{alg:train}
\begin{algorithmic}[1]
\Require training pairs $\mathcal{D}=\{(\mathbf{s},i_t)\}$, frozen retriever $R$, base LLM, validation-selected acceptance threshold $\tau$
\Statex \textbf{Stage 1: SFT with retrieval verification}
\State Initialize $\Qphi$ from the base LLM
\State Generate $M=5$ candidate queries for each training context
\State Retain each query with $\mathrm{rank}_R(i_t\mid q_j)\leq\tau$ \Comment{Eq.~\ref{eq:filter}}
\State Train LoRA adapters on accepted query targets for 3 epochs
\Statex \textbf{Stage 2: GRPO}
\State $\Qphizero\leftarrow\Qphi$ \Comment{freeze SFT reference}
\For{each policy optimization step}
  \State Select a training pair $(\mathbf{s},i_t)$
  \State Generate $G$ candidate query sets from $\Qphi(\cdot\mid\mathbf{s})$
  \For{$g=1$ to $G$}
    \State Retrieve with each query in $q^{(g)}$
    \State Fuse the retrieval lists by RRF, as at inference
    \State Compute $\rho_F(i_t\mid q^{(g)})$ and $r^{(g)}$ using Eq.~\ref{eq:reward}
  \EndFor
  \State Compute group advantages using Eq.~\ref{eq:advantage}
  \State Update $\phi$ once using Eq.~\ref{eq:grpo}
\EndFor
\end{algorithmic}
\end{algorithm}

\begin{algorithm}[h]
\caption{\name{} inference (one user)}
\label{alg:infer}
\begin{algorithmic}[1]
\Require user context $\su$, frozen retriever $R$, item index, trained $\Qphi$
\State Infer $K$ queries with the trained generator $\Qphi$
\State Form the query set $q=(q_1,\ldots,q_K)$
\State Retrieve a ranked list with each $q_j$, masking history items
\State Fuse the $K$ lists by RRF
\State \Return the $N$ highest ranked items
\end{algorithmic}
\end{algorithm}

%% file: sections/09_theory.tex
\Needspace{9\baselineskip}
\section{Analysis of the Retrieval Objective}
\label{sec:appendix_analysis}

We detail the analysis in Section~\ref{sec:recommendation_analysis} for a fixed training event and frozen retrieval system.

\subsection{Item-Supervised Query Verification}
\label{sec:analysis_verification}

Let $P_0(x\mid\mathbf{s})$ be the base proposal distribution over individual query strings, and define
\begin{equation}
  \mathcal{A}_{\mathbf{s},i_t}
  =\{x:\mathrm{rank}_R(i_t\mid x)\leq\tau\}.
\end{equation}
Write $a=P_0(\mathcal{A}_{\mathbf{s},i_t}\mid\mathbf{s})>0$.
Ideal rejection sampling yields the conditional distribution
\begin{equation}
  P_{\mathrm{acc}}(x\mid\mathbf{s},i_t)
  =\frac{P_0(x\mid\mathbf{s})\mathbf{1}[x\in\mathcal{A}_{\mathbf{s},i_t}]}{a}.
  \label{eq:accepted_distribution}
\end{equation}
For any distribution $P$ supported on accepted queries with finite $D_{\mathrm{KL}}(P\Vert P_0)$,
\begin{equation}
  D_{\mathrm{KL}}(P\Vert P_0)
  =D_{\mathrm{KL}}(P\Vert P_{\mathrm{acc}})-\log a.
  \label{eq:verification_projection}
\end{equation}
On the accepted support, $P_0(x\mid\mathbf{s})=aP_{\mathrm{acc}}(x\mid\mathbf{s},i_t)$.
Substitution into $\mathbb{E}_{x\sim P}\log(P(x)/P_0(x))$ gives Eq.~\ref{eq:verification_projection}.
The nonnegativity of KL then identifies $P_{\mathrm{acc}}$ as the closest such distribution to $P_0$ in $D_{\mathrm{KL}}(P\Vert P_0)$.

This standard conditional-distribution identity makes the role of the item label explicit: it defines which expressions are admissible, while their relative proposal probabilities are retained within the accepted set.
Our SFT uses finite accepted samples assembled into responses, so the identity describes query selection rather than an exact solution for the trained response distribution.
The identity assumes a nonzero acceptance probability.

\subsection{Query-Level Ranks and Fused Utility}
\label{sec:analysis_fusion}

Consider RRF~\citep{cormack2009rrf} over complete ranked lists, with a fixed constant $c>0$ and score
\begin{equation}
  S_q(i)=\sum_{j=1}^{K}\frac{1}{c+\mathrm{rank}_R(i\mid q_j)}.
\end{equation}
\begin{lemma}[Individual ranks do not determine fused rank]
\label{lem:rank_vector}
For any $c>0$, the target-rank vector alone does not determine the target's RRF rank over general collections of complete ranked lists.
\end{lemma}
\begin{proof}
For target $t$ and competing items $a,b$, consider two lists per set:
\begin{center}
\small
\setlength{\tabcolsep}{5pt}
\begin{tabular}{@{}ccccc@{}}
\toprule
Set & List 1 & List 2 & Target ranks & Fused rank \\
\midrule
$A$ & $(a,t,b)$ & $(a,t,b)$ & $(2,2)$ & $2$ \\
$B$ & $(a,t,b)$ & $(b,t,a)$ & $(2,2)$ & $3$ \\
\bottomrule
\end{tabular}
\end{center}
In $A$, item $a$ remains above $t$ and item $b$ remains below it after fusion.
In $B$, each competitor obtains one first and one third position, giving the score difference
\begin{align}
  S_B(a)-S_B(t)
  &=S_B(b)-S_B(t)\nonumber\\
  &=\frac{1}{c+1}+\frac{1}{c+3}-\frac{2}{c+2}\nonumber\\
  &=\frac{2}{(c+1)(c+2)(c+3)}>0.
\end{align}
Both competitors therefore outrank $t$ despite unchanged target ranks in the individual lists.
Thus, even the full target-rank vector, including its best and average ranks, is insufficient to determine fused utility.
\end{proof}
The construction does not imply that greater query diversity necessarily improves recommendations: useful combinations depend on their effects on both the target and competing items.

\subsection{Weighted Top-$N$ Miss Risk}
\label{sec:analysis_topn}

\begin{proof}[Proof of Proposition~\ref{prop:weighted_cutoff}]
For any realized rank $z\in\{1,\ldots,L\}$,
\begin{align}
  \log(1+z)-\log 2
  &=\sum_{n=1}^{z-1}\bigl[\log(n+2)-\log(n+1)\bigr]\nonumber\\
  &=\sum_{n=1}^{L-1}\log\frac{n+2}{n+1}\,\mathbf{1}[z>n].
\end{align}
Taking expectations proves Eq.~\ref{eq:topn_rank_risk}, since $\Pr(Z>n)=1-h_n$; the sum is empty when $L=1$.
\end{proof}
For events with different eligible-pool sizes $L_e$, use the common upper bound $L_{\max}=\max_e L_e$ and set $h_{e,n}=1$ for $n\geq L_e$.
The added terms are zero, so averaging the identity across events makes the averaged $h_{e,n}$ the expected Recall@$n$ in the single-target setting.
Each weight $w_n$ is positive and decreases with $n$.
Equivalently, improving a rank from $r+1$ to $r$ increases reward by $\log((r+2)/(r+1))$, giving larger per-position gains near the head.

The expansion characterizes the loss as a weighted combination of missed-target risks, rather than as the loss at any single cutoff.
A smaller expected log-rank therefore need not improve every individual Recall@$N$ or NDCG@$N$.
For truncated rank measurements with a finite missing-target sentinel, the same algebra applies to that measured rank variable; its interpretation is correspondingly limited to the recorded depths.

\subsection{Resolution of Group-Relative Feedback}
\label{sec:analysis_resolution}

\begin{lemma}[Resolution of centered advantages]
\label{lem:advantage_resolution}
For a fixed group of $G\geq1$ target ranks $z_1,\ldots,z_G$ and a cutoff $N$, let $b_g=\mathbf{1}[z_g\leq N]$ and $r_g=-\log(1+z_g)$.
If the centered binary-reward vector is nonzero, then the centered rank-reward vector is nonzero.
\end{lemma}
\begin{proof}
A centered vector is nonzero exactly when the underlying rewards are nonconstant.
Nonconstant binary rewards imply distinct ranks, which yield distinct log-rank rewards by strict monotonicity.
\end{proof}
The converse need not hold: different ranks on the same side of $N$ have equal binary rewards.
Thus, on any fixed collection of rollout groups, rank feedback preserves at least as many groups with nonzero centered advantages as binary hit feedback, without an independence assumption.

For a probabilistic expression, additionally suppose that the $G$ ranks are conditionally independent draws from $p_z=\Pr(Z=z)$ for the fixed event.
Let $p_{\mathrm{hit}}=\sum_{z\leq N}p_z$.
The probabilities of a constant-reward group are
\begin{align}
  P_{\mathrm{rank}}&=\sum_{z=1}^{L}p_z^G,\nonumber\\
  P_{\mathrm{hit}}&=p_{\mathrm{hit}}^G+(1-p_{\mathrm{hit}})^G
  \geq P_{\mathrm{rank}}.
\end{align}
The inequality follows by grouping the rank probabilities into the hit and miss sets and expanding their $G$th powers.
It holds with equality for $G=1$; for $G>1$, it is strict if either set contains at least two ranks of positive probability.

Nonzero centered advantages need not yield a nonzero parameter gradient: score-function terms may cancel, and KL contributes separately.
Any strictly monotonic rank reward shares this resolution property, while Appendix~\ref{sec:analysis_topn} specifies log-rank's weighting.

%% file: bib/refs.bib
@inproceedings{kang2018sasrec,
	title        = {Self-Attentive Sequential Recommendation},
	author       = {Kang, Wang-Cheng and McAuley, Julian},
	year         = 2018,
	booktitle    = {2018 {IEEE} International Conference on Data Mining ({ICDM})}
}

@inproceedings{he2020lightgcn,
	title        = {{LightGCN}: Simplifying and Powering Graph Convolution Network for Recommendation},
	author       = {He, Xiangnan and Deng, Kuan and Wang, Xiang and Li, Yan and Zhang, Yongdong and Wang, Meng},
	year         = 2020,
	booktitle    = {Proceedings of the 43rd International {ACM SIGIR} Conference on Research and Development in Information Retrieval}
}

@inproceedings{rajput2023tiger,
	title        = {Recommender Systems with Generative Retrieval},
	author       = {Rajput, Shashank and Mehta, Nikhil and Singh, Anima and Hulikal Keshavan, Raghunandan and Vu, Trung and Heldt, Lukasz and Hong, Lichan and Tay, Yi and Tran, Vinh and Samost, Jonah and Kula, Maciej and Chi, Ed and Sathiamoorthy, Maheswaran},
	year         = 2023,
	booktitle    = {Advances in Neural Information Processing Systems}
}

@inproceedings{hidasi2015gru4rec,
	title        = {Session-Based Recommendations with Recurrent Neural Networks},
	author       = {Hidasi, Bal{\'a}zs and Karatzoglou, Alexandros and Baltrunas, Linas and Tikk, Domonkos},
	year         = 2016,
	booktitle    = {International Conference on Learning Representations}
}

@inproceedings{sun2019bert4rec,
	title        = {{BERT4Rec}: Sequential Recommendation with Bidirectional Encoder Representations from Transformer},
	author       = {Sun, Fei and Liu, Jun and Wu, Jian and Pei, Changhua and Lin, Xiao and Ou, Wenwu and Jiang, Peng},
	year         = 2019,
	booktitle    = {Proceedings of the 28th {ACM} International Conference on Information and Knowledge Management}
}

@article{zhang2024instructrec,
	title        = {Recommendation as Instruction Following: A Large Language Model Empowered Recommendation Approach},
	author       = {Zhang, Junjie and Xie, Ruobing and Hou, Yupeng and Zhao, Wayne Xin and Lin, Leyu and Wen, Ji-Rong},
	year         = 2025,
	journal      = {ACM Transactions on Information Systems}
}

@inproceedings{reimers2019sbert,
	title        = {{Sentence-BERT}: Sentence Embeddings using Siamese {BERT}-Networks},
	author       = {Reimers, Nils and Gurevych, Iryna},
	year         = 2019,
	booktitle    = {Proceedings of the 2019 Conference on Empirical Methods in Natural Language Processing and the 9th International Joint Conference on Natural Language Processing ({EMNLP-IJCNLP})}
}

@article{qwen25,
	title        = {{Qwen2.5} Technical Report},
	author       = {{Qwen Team}},
	year         = 2024,
	journal      = {arXiv preprint arXiv:2412.15115}
}

@inproceedings{cormack2009rrf,
	title        = {Reciprocal Rank Fusion Outperforms {Condorcet} and Individual Rank Learning Methods},
	author       = {Cormack, Gordon V. and Clarke, Charles L. A. and Buettcher, Stefan},
	year         = 2009,
	booktitle    = {Proceedings of the 32nd International {ACM SIGIR} Conference on Research and Development in Information Retrieval}
}

@article{shao2024deepseekmath,
	title        = {{DeepSeekMath}: Pushing the Limits of Mathematical Reasoning in Open Language Models},
	author       = {Shao, Zhihong and Wang, Peiyi and Zhu, Qihao and Xu, Runxin and Song, Junxiao and Bi, Xiao and Zhang, Haowei and Zhang, Mingchuan and Li, Y. K. and Wu, Y. and Guo, Daya},
	year         = 2024,
	journal      = {arXiv preprint arXiv:2402.03300}
}

@inproceedings{xu2025iagent,
	title        = {{iAgent}: {LLM} Agent as a Shield between User and Recommender Systems},
	author       = {Xu, Wujiang and Shi, Yunxiao and Liang, Zujie and Ning, Xuying and Mei, Kai and Wang, Kun and Zhu, Xi and Xu, Min and Zhang, Yongfeng},
	year         = 2025,
	booktitle    = {Findings of the Association for Computational Linguistics: ACL 2025}
}

@inproceedings{han2025querec,
  title={Rethinking LLM-Based Recommendations: A Personalized Query-Driven Parallel Integration.},
  author={Han, Donghee and Song, Hwanjun and Yi, Mun Yong},
	year         = 2025,
	booktitle    = {Findings of the Association for Computational Linguistics: EMNLP 2025}
}

@article{queryagent2026,
	title        = {{QueryAgent-R1}: Bridging Query Generation and Product Retrieval for {E-Commerce} Query Recommendation},
	author       = {Sun, Dike and Zou, Zheng and Zang, Jingtong and Sun, Qi and Zhao, Huaipeng and Luo, Tao and Zeng, Xiaoyi},
	year         = 2026,
	journal      = {arXiv preprint arXiv:2606.05671}
}

@inproceedings{memrec2026,
	title        = {{MemRec}: Collaborative Memory-Augmented Agentic Recommender System},
	author       = {Chen, Weixin and Zhao, Yuhan and Huang, Jingyuan and Ye, Zihe and Ju, Mingxuan and Zhao, Tong and Shah, Neil and Chen, Li and Zhang, Yongfeng},
	year         = 2026,
	booktitle    = {Proceedings of the 64th Annual Meeting of the Association for Computational Linguistics (Volume 1: Long Papers)}
}

@inproceedings{hou2022unisrec,
	title        = {Towards Universal Sequence Representation Learning for Recommender Systems},
	author       = {Hou, Yupeng and Mu, Shanlei and Zhao, Wayne Xin and Li, Yaliang and Ding, Bolin and Wen, Ji-Rong},
	year         = 2022,
	booktitle    = {Proceedings of the 28th {ACM SIGKDD} Conference on Knowledge Discovery and Data Mining}
}

@inproceedings{letter2024,
	title        = {Learnable Item Tokenization for Generative Recommendation},
	author       = {Wang, Wenjie and Bao, Honghui and Lin, Xinyu and Zhang, Jizhi and Li, Yongqi and Feng, Fuli and Ng, See-Kiong and Chua, Tat-Seng},
	year         = 2024,
	booktitle    = {Proceedings of the 33rd ACM International Conference on Information and Knowledge Management}
}

@article{onerec2024,
  title={Onerec: Unifying retrieve and rank with generative recommender and iterative preference alignment},
  author={Deng, Jiaxin and Wang, Shiyao and Cai, Kuo and Ren, Lejian and Hu, Qigen and Ding, Weifeng and Luo, Qiang and Zhou, Guorui},
  journal={arXiv preprint arXiv:2502.18965},
  year={2025}
}

@inproceedings{li2023gpt4rec,
	title        = {{GPT4Rec}: A Generative Framework for Personalized Recommendation and User Interests Interpretation},
	author       = {Li, Jinming and Zhang, Wentao and Wang, Tian and Xiong, Guanglei and Lu, Alan and Medioni, Gerard},
	year         = 2023,
	booktitle    = {Proceedings of the 2023 SIGIR Workshop on eCommerce},
	series       = {CEUR Workshop Proceedings}
}

@article{bao2023bigrec,
  title={A bi-step grounding paradigm for large language models in recommendation systems},
  author={Bao, Keqin and Zhang, Jizhi and Wang, Wenjie and Zhang, Yang and Yang, Zhengyi and Luo, Yanchen and Chen, Chong and Feng, Fuli and Tian, Qi},
  journal={ACM Transactions on Recommender Systems},
  year={2025},
}

@inproceedings{li2023recformer,
	title        = {Text Is All You Need: Learning Language Representations for Sequential Recommendation},
	author       = {Li, Jiacheng and Wang, Ming and Li, Jin and Fu, Jinmiao and Shen, Xin and Shang, Jingbo and McAuley, Julian},
	year         = 2023,
	booktitle    = {Proceedings of the 29th ACM SIGKDD Conference on Knowledge Discovery and Data Mining}
}

@inproceedings{zhang2024agentcf,
	title        = {{AgentCF}: Collaborative Learning with Autonomous Language Agents for Recommender Systems},
	author       = {Zhang, Junjie and Hou, Yupeng and Xie, Ruobing and Sun, Wenqi and McAuley, Julian and Zhao, Wayne Xin and Lin, Leyu and Wen, Ji-Rong},
	year         = 2024,
	booktitle    = {Proceedings of the {ACM} Web Conference 2024}
}

@inproceedings{geng2022p5,
	title        = {Recommendation as Language Processing ({RLP}): A Unified Pretrain, Personalized Prompt \& Predict Paradigm ({P5})},
	author       = {Geng, Shijie and Liu, Shuchang and Fu, Zuohui and Ge, Yingqiang and Zhang, Yongfeng},
	year         = 2022,
	booktitle    = {Proceedings of the 16th {ACM} Conference on Recommender Systems}
}

@inproceedings{hu2022lora,
	title        = {{LoRA}: Low-Rank Adaptation of Large Language Models},
	author       = {Hu, Edward J. and Shen, Yelong and Wallis, Phillip and Allen-Zhu, Zeyuan and Li, Yuanzhi and Wang, Shean and Wang, Lu and Chen, Weizhu},
	year         = 2022,
	booktitle    = {International Conference on Learning Representations}
}

@article{robertson2009bm25,
	title        = {The Probabilistic Relevance Framework: {BM25} and Beyond},
	author       = {Robertson, Stephen and Zaragoza, Hugo},
	year         = 2009,
	journal      = {Foundations and Trends in Information Retrieval}
}

@inproceedings{xu2011semrec,
	title        = {{SemRec}: A Semantic Enhancement Framework for Tag Based Recommendation},
	author       = {Xu, Guandong and Gu, Yanhui and Dolog, Peter and Zhang, Yanchun and Kitsuregawa, Masaru},
	year         = 2011,
	booktitle    = {Proceedings of the AAAI Conference on Artificial Intelligence}
}

@inproceedings{li2019mind,
	title        = {Multi-Interest Network with Dynamic Routing for Recommendation at {Tmall}},
	author       = {Li, Chao and Liu, Zhiyuan and Wu, Mengmeng and Xu, Yuchi and Zhao, Huan and Huang, Pipei and Kang, Guoliang and Chen, Qiwei and Li, Wei and Lee, Dik Lun},
	year         = 2019,
	booktitle    = {Proceedings of the 28th ACM International Conference on Information and Knowledge Management}
}

@inproceedings{yu2025tagcf,
	title        = {Who You Are Matters: Bridging Interests and Social Roles via {LLM}-Enhanced Logic Recommendation},
	author       = {Yu, Qing and Wang, Xiaobei and Liu, Shuchang and Bai, Yandong and Yang, Xiaoyu and Wang, Xueliang and Meng, Chang and Wu, Shanshan and Yang, Hailan and Wen, Bin and Xiao, Huihui and Li, Xiang and Yang, Fan and Feng, Xiaoqiang and Hu, Lantao and Li, Han and Gai, Kun and Zou, Lixin},
	year         = 2025,
	booktitle    = {Advances in Neural Information Processing Systems}
}

@inproceedings{wu2026guider,
	title        = {{GUIDER}: Generative User Interest Discovery via Explicit Reasoning with Large Language Models},
	author       = {Wu, Jinke and Wu, Yinghao and Liu, Shuchang and Bai, Yandong and Zhu, Qi and Xiao, Huihui and Yang, Fan},
	year         = 2026,
	booktitle    = {Proceedings of the 32nd ACM SIGKDD Conference on Knowledge Discovery and Data Mining V.2}
}

@article{mistral7b,
	title        = {{Mistral 7B}},
	author       = {Jiang, Albert Q. and Sablayrolles, Alexandre and Mensch, Arthur and Bamford, Chris and Chaplot, Devendra Singh and de las Casas, Diego and Bressand, Florian and Lengyel, Gianna and Lample, Guillaume and Saulnier, Lucile and Renard Lavaud, L{\'e}lio and Lachaux, Marie-Anne and Stock, Pierre and Le Scao, Teven and Lavril, Thibaut and Wang, Thomas and Lacroix, Timoth{\'e}e and El Sayed, William},
	year         = 2023,
	journal      = {arXiv preprint arXiv:2310.06825}
}

@article{llama3,
	title        = {The {Llama 3} Herd of Models},
	author       = {{Llama Team}},
	year         = 2024,
	journal      = {arXiv preprint arXiv:2407.21783}
}

@article{lin2025recr1,
	title        = {{Rec-R1}: Bridging Generative Large Language Models and User-Centric Recommendation Systems via Reinforcement Learning},
	author       = {Lin, Jiacheng and Wang, Tian and Qian, Kun},
	year         = 2025,
	journal      = {Transactions on Machine Learning Research}
}

@inproceedings{tao2026personalizing,
  title={Personalizing How Recommendation Agents Reason: Self-Evolving Per-User Skills for Recommendation Agents},
  author={Tao, Zhen and Lai, Riwei and Yu, Chenyun and Chen, Weixin and Chen, Li and Kong, Beibei and Cheng, Lei and Zhuo, Chengxiang and Li, Zang and Sun, Qingqiang},
  booktitle={Findings of the Association for Computational Linguistics: EMNLP 2026},
  year={2026}
}

@inproceedings{kang2026personalized,
  title={Personalized News Headline Generation via Multi-Agent User Simulation},
  author={Kang, Li and Zhao, Yuhan and Chen, Li},
  booktitle={Findings of the Association for Computational Linguistics: EMNLP 2026},
  year={2026}
}

@inproceedings{sun2026diagnosing,
  title={Diagnosing LLM Arbitration Behavior over Pre-evidence Epistemic States in RAG-based Fact-Checking},
  author={Sun, Yuxi and Shang, Wenbo and Gao, Wei and Huang, Xin and Ma, Jing},
  booktitle={Proceedings of the 2026 Conference on Empirical Methods in Natural Language Processing},
  year={2026}
}
